\documentclass[12pt,a4paper]{article}

\usepackage{natbib}
\bibpunct[:]{(}{)}{,}{a}{}{,}
\usepackage{graphicx}
\usepackage{bm}
\usepackage{amsmath,amsthm}
\usepackage{amssymb}
\usepackage{bbm}
\usepackage{color}
\usepackage{comment}
\usepackage{algorithm}
\usepackage{algpseudocode}
\usepackage{authblk}

\DeclareMathOperator*{\argmax}{\mathop{\rm argmax}\limits}
\DeclareMathOperator*{\argmin}{\mathop{\rm argmin}\limits}
\newcommand{\bs}{\boldsymbol}
\newcommand{\mb}{\mathbf}
\newtheorem{cor}{Corollary}
\newtheorem*{rem}{Remark}
\newtheorem{thm}{Theorem}
\newtheorem{prop}{Proposition}

\allowdisplaybreaks

\title{New Insight of Spatial Scan Statistics via Regression Model}
\date{}

\author[1]{Takayuki Kawashima} \affil[1]{ Department of Mathematical and Computing Science, Institute of Science Tokyo, Tokyo, Japan 
\authorcr E-mail: \texttt{kawashima@c.titech.ac.jp} }

\author[2]{Daisuke Yoneoka} \affil[2]{Center for Surveillance, Immunization, and Epidemiologic Research, National Institute of Infectious Diseases,  Tokyo, Japan}

\author[3]{Yuta Tanoue} \affil[3]{ Faculty of Marine Technology, Tokyo University of Marine Science and Technology, Tokyo, Japan }

\author[4]{Akifumi Eguchi} \affil[4]{ Center for Preventive Medical Sciences, Chiba University, Chiba, Japan }

\author[5]{Shuhei Nomura} \affil[5]{Department of Health Policy and Management, School of Medicine, Keio University, Tokyo, Japan}

\begin{document}

\maketitle

\begin{abstract}
The spatial scan statistic is widely used to detect disease clusters in epidemiological surveillance. 
Since the seminal work by~\cite{kulldorff1997}, numerous extensions have emerged, including methods for defining scan regions, detecting multiple clusters, and expanding statistical models. 
Notably,~\cite{jung2009} and~\cite{ZHANG20092851} introduced a regression-based approach accounting for covariates, encompassing classical methods such as those of~\cite{kulldorff1997}. 
Another key extension is the expectation-based approach~\citep{neill2005anomalous,neillphdthesis}, which differs from the population-based approach represented by~\cite{kulldorff1997} in terms of hypothesis testing. 
In this paper, we bridge the regression-based approach with both expectation-based and population-based approaches. 
We reveal that the two approaches are separated by a simple difference: the presence or absence of an intercept term in the regression model. 
Exploiting the above simple difference, we propose new spatial scan statistics under the Gaussian and Bernoulli models. 
We further extend the regression-based approach by incorporating the well-known sparse L0 penalty and show that the derivation of spatial scan statistics can be expressed as an equivalent optimization problem. 
Our extended framework accommodates extensions such as space-time scan statistics and detecting multiple clusters while naturally connecting with existing spatial regression-based cluster detection. 
Considering the relation to case-specific models~\citep{she2011,10.1214/11-STS377}, clusters detected by spatial scan statistics can be viewed as outliers in terms of robust statistics. 
Numerical experiments with real data illustrate the behavior of our proposed statistics under specified settings. 
\end{abstract}

%%%Sect1. Intro %%%
\section{Introduction}\label{sect intro}
%%%
Spatial scan statistics are widely used in epidemiology to identify disease clusters and in other fields to detect spatial clusters. 
Seminal work on spatial scan statistics was conducted by~\cite{kull1995} and~\cite{kulldorff1997}. 
This approach relies on the likelihood ratio, selecting the cluster with the highest likelihood ratio among predefined potential clusters as the test statistic and calculating p-values using Monte Carlo hypothesis testing~\citep{10.1214/aoms/1177707045}. 
Due to their effectiveness, spatial scan statistics have been applied to analyze sudden infant death syndrome~\citep{kulldorff1997}, leukemia~\citep{kull1995, lukemia2007}, and lung cancer~\citep{lungcancer1999}. 
More recently, they have been employed in several studies focusing on COVID-19~\citep{covid1, covid2, covid3, covid4, covid5}. 

Various extensions to spatial scan statistics have been proposed, including alternative methods for constructing potential clusters, such as elliptical shapes~\citep{https://doi.org/10.1002/sim.2490}, flexible shapes~\citep{tango2005flexibly}, and public network-based shapes~\citep{tanoue2023publictrans}; alternative statistical models, such as a Bayesian framework~\citep{NIPS2005_28acfe2d} and kernel methods~\citep{10.1145/3347146.3359101}; and adaptations for various types of data, including censored data~\citep{10.1111/j.1541-0420.2006.00661.x, 10.1111/j.1541-0420.2006.00714.x}, ordinal data~\citep{LEE201928}, and spatio-temporal data~\citep{kull2001}. 
For a comprehensive overview, see the extensive survey on spatial scan statistics~\citep{10.1214/21-SS132}.

One significant advancement in this field has been the incorporation of regression models to handle covariates~\citep{jung2009, ZHANG20092851}. 
This approach accommodates covariates while including Kulldorff's original spatial scan statistics as a special case when no covariates are considered (see Sect. 2.2-3 in~\cite{jung2009}). 
It serves as a natural generalization of Kulldorff's method. 
By leveraging regression models, many advanced methods with modifications~\citep{https://doi.org/10.1111/j.1541-0420.2008.01069.x, https://doi.org/10.1111/biom.12509, https://doi.org/10.1002/sim.8459} can be unified within the framework of spatial scan statistics. 

Another notable development involves the expectation-based approach, which employs a different hypothesis testing strategy compared to Kulldorff's population-based method~\citep{neill2005anomalous,neillphdthesis}. 
This approach has shown superior performance over traditional methods in certain scenarios, depending on the data type and quality. 
However, the theoretical relationship between these approaches and their connection to regression models remains unclear. 
To date, each method has advanced through independent research. 

In this paper, we first reveal the theoretical connections between the population-based and expectation-based approaches within the framework of spatial scan statistics as a regression model. 
Using this framework, we propose new spatial scan statistics for the Gaussian and Bernoulli models. 
Second, we introduce an equivalent optimization problem for the spatial scan statistics framework by integrating a simple extension of the regression model proposed by~\cite{jung2009, ZHANG20092851} with the well-known sparse $\ell_0$ penalty. 
One advantage of the proposed approach is that it enables the use of various techniques from the field of sparse regression, such as selective inference~\citep{lee2016exact, tibshirani2016exact}, eliminating the need for Monte Carlo hypothesis testing to compute p-values. 
We also demonstrate how the proposed framework can be extended to existing enhancements of spatial scan statistics, including space-time scan statistics~\citep{kull2001} and multiple clusters detection methods~\citep{Zhang2010SpatialSS, li2011}. 
Furthermore, beyond spatial scan statistics, our framework connects with established methodologies for spatial regression-based cluster detection and robust statistics. 

The remainder of this paper is organized as follows: In Sect. 2, we briefly review spatial scan statistics with covariates based on the regression model, provide a description of expectation-based and population-based approaches, and show connections to existing methods. 
Then, we propose new spatial scan statistics by taking advantage of the expression in the regression model. 
In Sect. 3, we present the equivalent framework via $\ell_0$ penalized regression for spatial scan statistics and investigate its relation to spatial regression methods for cluster detection and robust statistics. 
In Sect. 4, we provide numerical experiments, which we apply to real data to examine our proposed spatial scan statistics. 
The conclusion is presented in Sect. 5.

%%%Sect.2 %%%
\section{Spatial Scan Statistics with Covariates}\label{sect sp scan brief}
Suppose that we have $N \ (i = 1, \ldots, N)$ regions, each of which has covariate variables $ x_{i} \in \mathbb{R}^p$ and an outcome $y_{i} \in \mathbb{R}$ (typically, the number of cases of a disease). 
Let $K$ be the number of potential clusters. 
Let $S_{k} \ (k=1, \ldots, K)$ denote the $k$-th potential cluster which consists of some regions and $Z_{k,i}$ be an indicator variable which means $Z_{k,i}=1 \mbox{ if } i \in S_{k}$, and $0$ otherwise, e.g., $S_1 = \left\{ 1,2,4 \right\}, Z_{1,1} =1, \mbox{ and } Z_{1,3} = 0  $. 
The potential cluster $S_{k}$ is ordinarily constructed using spatial information such as longitude and latitude.  
There are many studies on constructing potential clusters, such as rectangle shaped~\citep{naus1965clustering,loader1991large}, circular shaped~\citep{kulldorff1997}, elliptic shaped~\citep{https://doi.org/10.1002/sim.2490}, and flexibly shaped ones~\citep{tango2005flexibly}. 
We do not focus on how to construct potential clusters and assume that potential clusters are given in advance. 

\subsection{Regression-Based Spatial Scan Statistic}\label{subsect1 sp scan brief}
\cite{ZHANG20092851} and \cite{jung2009} proposed the spatial scan statistic based on regression models to treat covariates. 
They consider the conditional probability density function $ f (y_{i} | x_{i}, Z_{k,i}; \mu_i) $ with a mean parameter $\mu_i$. 
For a regression model, they assume the following linear regression model: 
\begin{align}\label{reg model}
g( \mathrm{E}_{ f (y_{i} | x_{i}, Z_{k,i}; \alpha, \theta_{S_{k}}, \bs{\beta}^{\top} ) } [ y_{i} ] )= g ( \mu_i ) =  \alpha + \theta_{S_{k}} Z_{k ,i} + x_{i}^{\top} \bs{\beta} ,
\end{align}
where $g$ is a monotonically differentiable link function, and $(\alpha, \theta_{S_{k}}, \bs{\beta}^{\top} )$ are the regression coefficients. 
$\theta_{S_{k}}$ expresses the additive effect of the potential cluster $S_{k}$. 
In what follows, we use $( \alpha, \theta_{S_k}, \bs{\beta}^{\top} )$ as parameters of the regression model instead of the mean parameter $\mu_i$ for simplicity. 

The following null and alternative hypotheses are considered. 
\begin{align*}
H_{0} &: \theta_{S_{k}} =0, \\
H_{1} &: \theta_{S_{k}} \neq 0.
\end{align*}
We prepare the parameter space $(\alpha, \theta_{S_{k}}, \bs{\beta}^{\top} ) \in \Theta \subset ( \mathbb{R} \times \mathbb{R} \times \mathbb{R}^p  )$ and $\Theta_{0}  ( \subset \Theta ) $, where $ \theta_{S_K} = 0$ for the null model. 
By $\Theta$ and $\Theta_0$, the log-likelihood ratio test statistic for $S_{k}$ is given by 
\begin{align}\label{llr def}
\mathrm{LLR}_k =  \log \frac{ \max_{ (\alpha, \theta_{S_{k}}, \bs{\beta}^{\top}   ) \in \Theta }  \Pi_{i=1}^N  f (y_{i} | x_{i}, Z_{k,i}; \alpha, \theta_{S_{k}}, \bs{ \beta }^{\top}  ) }{   \max_{ (\alpha, 0, \bs{\beta}^{\top}  ) \in \Theta_{0} }  \Pi_{i=1}^N   
f (y_{i} | x_{i}, Z_{k,i}; \alpha, 0, \bs{ \beta }^{\top}  ) }.
\end{align}
We define $S_0 \mbox{ and the corresponding log-likelihood ratio } \mathrm{LLR}_0 = 0$. 
The most likely cluster (MLC) is defined as the potential cluster that can achieve the maximum value of \eqref{llr def} over $\left\{ S _0, S_1, \ldots, S_K \right\}$, as follows: 
\begin{align}\label{llr max}
\mbox{MLC is } S_{k^*}, \mbox{ where } k^{*} = \argmax_{ k \in  \left\{ 0, \ldots, K \right\} } \mathrm{LLR}_k.
\end{align}
When the MLC is $S_0$, it indicates that the MLC is the empty set, i.e., there are no clusters within potential clusters. 
To evaluate the significance of the MLC except for the case of $k^* = 0$, one can compute the p-value of the MLC generally via the procedure of Monte Carlo hypothesis testing~\citep{10.1214/aoms/1177707045}. 

\subsection{Connection to Existing Spatial Scan Statistics}\label{subsect2 sp scan brief}
Here, we show that regression-based spatial scan statistics without covariates encompass various existing spatial scan statistics, including Kulldorff's spatial scan statistics. 
In addition, because the regression model can be flexibly transformed according to the purpose, e.g., interaction terms, the distribution of the outcome, and so on, we can propose a new type of spatial scan statistic that has not yet been suggested. 
While covariates are excluded from the regression model here to enable comparison with existing spatial scan statistics, they can be incorporated into all the proposed spatial scan statistics below. 

\subsubsection{Population-based and Expectation-based Spatial Scan Statistics}
For simplicity of understanding, we consider the following statistical model for the spatial scan statistics. 
\begin{align*}
c_i \sim g ( q_i b_i) \ \mbox{for }i=1,\ldots,N,
\end{align*}
where $c_i$ is a count, $g$ is a probability density function with the mean $q_i b_i$, and $q_i$ and $b_i $ are an unknown rate parameter and the baseline, respectively. 
For spatial scan statistics, one wants to detect regions that show significantly higher or lower $q_i$ via hypothesis testing. 
\cite{neill2005anomalous} and~\cite{neillphdthesis} showed that there are two typical approaches, such as the population-based and expectation-based, in the spatial scan statistics in terms of the interpretation of the baseline $b_i$. 
According to~\cite{neillphdthesis}, both approaches are briefly explained below. 

In the population-based approach, the baseline $b_i$ is typically the population, e.g., from census data, and the rate $q_i$ is the underlying disease rate. 
In the framework of hypothesis testing, the population-based approach can be expressed as follows: 
\begin{align*}
& H_0 : q_i = q_{all} ,  \\
& H_1 : q_i = q_{in} \mbox{ inside some region } S,  \ q_i = q_{out} \mbox{ outside } S ,
\end{align*}
%\textbf{}
where $q_{all}, q_{in}, q_{out}$ are some constants, and $q_{in} \neq q_{out}$. 
As an example, Kulldorf's spatial scan statistics for the Poisson model and Bernoulli model~\citep{kull1995,kulldorff1997} 
are classified as population-based spatial scan statistics. 
In the expectation-based approach, the baseline $b_i$ is typically the expected count of the inferred case of disease, and the rate $q_i$ is the underlying relative risk. 
The framework of hypothesis testing for the expectation-based approach can be expressed as follows: 
\begin{align*}
& H_0 : q_i = 1 ,  \\
& H_1 : q_i = q_{in} \mbox{ inside some region } S,  \ q_i = 1 \mbox{ outside } S ,
\end{align*}
where $q_{in} \neq 1$. 

Although in~\cite{neillphdthesis}, a qualitative explanation, e.g., depends on the type of data, has been provided regarding which approach to use, no theoretical distinctions other than the differences in hypotheses have been addressed. 
Furthermore, the relationship with regression-based spatial scan statistics has not been explored. 
Therefore, we aim to clarify the differences between the population-based and expectation-based approaches from the perspective of regression-based spatial scan statistics. 
As a result, leveraging regression modeling's flexibility, we can verify and propose both approaches for various types of statistical models, including those already known. 
The following section presents the results for specific statistical models. 

\subsubsection{Case of Poisson Model}
We consider the following Poisson regression model as~\eqref{reg model}: 
\begin{align}\label{poisson without cov}
f (y_{i} | Z_{k,i}; \alpha, \theta_{S_k} ) = \frac{ \exp( -\mu_i ) }{ y_i !} \mu_i^{y_i}, 
\end{align}
where $ \log ( \mu_i ) = \alpha + \theta_{S_k} Z_{k,i} + \log (\gamma_i)  $, and $\gamma_i$ is an offset term (typically, a population at risk or expected counts of disease in the $i$th region). 
\cite{jung2009} demonstrated that the log-likelihood ratio~\eqref{llr def} under the Poisson regression model specified in equation~\eqref{poisson without cov} is equivalent to the following Kulldorff's spatial scan statistic~\citep{kulldorff1997}, which belongs to the category of population-based spatial scan statistics.  
\begin{align}\label{population based poisson kulldorf}
& \mathrm{LLR}_k \nonumber \\ 
& = \left\{ \left(  \sum_{i=1}^N y_i Z_{k,i} \right) \log \left( \frac{ \sum_{i=1}^N y_i Z_{k,i} }{  \sum_{i=1}^N \gamma_i Z_{k,i} } \right) + \left( \sum_{i=1}^N y_i ( 1 -Z_{k,i} )  \right) \log \left( \frac{ \sum_{i=1}^N y_i ( 1 - Z_{k,i} ) }{  \sum_{i=1}^N \gamma_i (1 - Z_{k,i}) } \right) \right\}  \nonumber \\
& \qquad - \left( \sum_{i=1}^N y_i \right) \log \left(  \frac{ \sum_{i=1}^N y_i }{ \sum_{i=1}^N \gamma_i } \right). 
\end{align}
Then, we can see that the log-likelihood ratio~\eqref{llr def} corresponds to Neill's spatial scan statistic~\citep{neill2005anomalous}, which falls under the category of expectation-based spatial scan statistics, by using the Poisson regression model \textit{without the intercept term $\alpha$}. 

\begin{prop}\label{prop equiv exp based under poisson}
We consider the following Poisson regression model without the intercept term $\alpha$ in~\eqref{poisson without cov}. 
\begin{align*}
f (y_{i} | Z_{k,i}; 0, \theta_{S_k} ) = \frac{ \exp( -\mu_i ) }{ y_i !} \mu_i^{y_i}, 
\end{align*}
where $ \log ( \mu_i ) =  \theta_{S_k} Z_{k,i} + \log (\gamma_i)  $. 
Then, the log-likelihood ratio~\eqref{llr def} is equivalent to the following Neill's spatial scan statistic~\citep{neill2005anomalous}. 
\begin{align}\label{expectation based poisson}
\mathrm{LLR}_k = \left( \sum_{i=1}^N y_i Z_{k,i} \right) \log \left( \frac{ \sum_{i=1}^N y_i Z_{k,i} }{  \sum_{i=1}^N \gamma_i Z_{k,i} } \right) + \sum_{i=1}^N ( \gamma_i - y_i )Z_{k,i}  .
\end{align}
\end{prop}
The proof is in the Appendix. 
The difference between~\eqref{population based poisson kulldorf} an~\eqref{expectation based poisson} is simply whether the regression model contains the intercept term or not. 
This difference can be seen as giving the following interpretation of the mean structure in the regression model. 
\begin{align*}
\frac{\mu_i }{ \gamma_i } = \exp(\alpha ) \times \exp ( \theta_{S_k} Z_{k,i} ).
\end{align*}
The left side shows the underlying disease rate and relative risk for the population-based and expectation-based spatial scan statistics, respectively. 
The right side is $\exp(\alpha)$ and $\exp(0) = 1$ for population-based and expectation-based spatial scan statistics, respectively, under no clusters, i.e., $\theta_{S_k} =0$. 
\begin{rem}
The theoretical difference between population-based and expectation-based approaches has not been deeply investigated. 
Recently,~\cite{tanoue2023publictrans} proved that under the conditions that $\sum_{i=1}^N y_i = \sum_{i=1}^N \gamma_i$ and $\sum_{i=1}^N y_i \rightarrow \infty$, the population-based~\eqref{population based poisson kulldorf} and expectation-based~\eqref{expectation based poisson} approaces are equivalent. 
However, it is for a case in terms of limit, making it difficult to apply the concept to other statistical models and making it less versatile. 
On the other hand, we can present a clear difference in the regression model, which is simply whether they include an intercept term or not. 
Subsequently, this relationship can be expanded to other statistical models and can be utilized to propose new spatial scan statistics. 
\end{rem}

\subsubsection{Case of Gaussian Model}\label{connect ex scan gauss model}
\cite{neillphdthesis} has proposed both population-based and expectation-based spatial scan statistics under the Gaussian model. 
Here, we can show that both spatial scan statistics proposed by~\cite{neillphdthesis} can be expressed to the regression-based spatial scan statistics, similarly to the Poisson model. 
\begin{prop}\label{prop equiv exp based under gauss}
We consider the following regression model under the Gaussian model:  
\begin{align*}
f (y_{i} | Z_{k,i}; \alpha, \theta_{S_k}  ) =  \frac{1}{ \sqrt{2 \pi \sigma_i^2 } } \exp \left\{ - \frac{ ( y_i- \mu_i )^2 }{ 2 \sigma_i^2 } \right\}  , 
\end{align*}
where $ \mu_i  =  \gamma_i( 1 + \alpha + \theta_{S_k} Z_{k,i} )$, and both $\gamma_i$ and $ \sigma_{i}^2$ are not parameters but pre-determined values. 
Then, the log-likelihood ratio~\eqref{llr def} is equivalent to the following population-based spatial scan statistic: 
\begin{align*}
\mathrm{LLR}_k = \frac{ \left( \sum_{i=1}^N \frac{y_i \gamma_i }{ \sigma_i^2} Z_{k,i} \right)^2 }{ 2 \sum_{i=1}^N \frac{ \gamma_i^2 }{ \sigma_{i}^2 } Z_{k,i}  } + \frac{ \left( \sum_{i=1}^N \frac{y_i \gamma_i }{ \sigma_i^2} ( 1 - Z_{k,i} ) \right)^2 }{ 2 \sum_{i=1}^N \frac{ \gamma_i^2 }{ \sigma_{i}^2 } ( 1 - Z_{k,i} )  } -\frac{ \left( \sum_{i=1}^N \frac{y_i \gamma_i }{ \sigma_i^2}  \right)^2 }{ 2 \sum_{i=1}^N \frac{ \gamma_i^2 }{ \sigma_{i}^2 }  }.
\end{align*}

In addition, we consider a regression model without the intercept term $\alpha$ under the Gaussian model. 
\begin{align*}
f (y_{i} | Z_{k,i}; 0, \theta_{S_k}  ) = \frac{1}{ \sqrt{2 \pi \sigma_i^2 } } \exp \left\{ - \frac{ ( y_i- \mu_i )^2 }{ 2 \sigma_i^2 } \right\}, 
\end{align*}
where $ \mu_i  =  \gamma_i ( 1 + \theta_{S_k} Z_{k,i} )$. 
Then, the log-likelihood ratio~\eqref{llr def} is equivalent to the following expectation-based spatial scan statistic: 
\begin{align*}
\mathrm{LLR}_k =  \frac{ \left( \sum_{i=1}^N \frac{ y_i \gamma_i }{ \sigma_i^2 }  Z_{k,i} \right)^2 }{ 2 \sum_{i=1}^N \frac{ \gamma_i^2 }{ \sigma_i^2 } Z_{k,i} } + \frac{ \sum_{i=1}^N \frac{ \gamma_i^2 }{ \sigma_i^2 } Z_{k,i} }{2} - \sum_{i=1}^N \frac{ y_i \gamma_i }{ \sigma_i^2 }  Z_{k,i}  .
\end{align*}
\end{prop}
The proof is in the Appendix. 
Again, the difference between population-based and expectation-based can be seen in whether the regression model includes the intercept term $\alpha$ or not. 
Similar to the Poisson model, the mean structure in the regression model can be interpreted as follows: 
\begin{align*}
\frac{\mu_i }{ \gamma_i } = 1 + \alpha  +  \theta_{S_k} Z_{k,i} .
\end{align*}
The right side shows that the baseline value of the ratio is $ 1 + \alpha $ for population-based and $ 1 + 0 = 1 $ for expectation-based under no clusters, i.e., $\theta_{S_k} =0$. 

For the fully parameterized case, i.e., the variance $\sigma^2$ is not given, \cite{Kulldorff2009} proposed another spatial scan statistic. 
We can obtain the same spatial scan statistic proposed by \cite{Kulldorff2009} via the regression-based spatial scan statistic as follows: 
\begin{prop}\label{prop population based full para gauss SSS}
We consider a regression model under the Gaussian model. 
\begin{align*}
f (y_{i} | Z_{k,i}; \alpha, \theta_{S_k} , \sigma^2  ) = \frac{1}{ \sqrt{2 \pi \sigma^2 } } \exp \left\{ - \frac{ ( y_i- \mu_i )^2 }{ 2 \sigma^2 } \right\}, 
\end{align*}
where $ \mu_i  =  \alpha  + \theta_{S_k} Z_{k,i} $. 
Here, the variance $\sigma^2$ is not given but a parameter under the Gaussian model. 
Then, the log-likelihood ratio~\eqref{llr def} is equivalent to the following spatial scan statistic proposed by \cite{Kulldorff2009}: 
\begin{align*}
\mathrm{LLR}_k = \frac{N}{2} \log \left\{ \frac{ 1 }{ N } \sum_{i=1}^N  ( y_i - \frac{1}{N} \sum_{i=1}^N y_i )^2 \right\}  - \frac{N}{2} \log \widehat{ \sigma_{S_k}^2 } ,
\end{align*}
where $ \widehat{  \sigma_{S_k}^2 } = \frac{1}{N} \left( \sum_{i=i}^N y_i^2   - \frac{ ( \sum_{i=1}^N y_i Z_{k,i}  )^2 }{ \sum_{i=1}^N Z_{k,i} } - \frac{ \left\{ \sum_{i=1}^N y_i (1- Z_{k,i} ) \right\}^2 }{ \sum_{i=1}^N (1 - Z_{k,i}) } \right) $.
\end{prop}
The proof is in the Appendix. 
Due to the expression of the regression form for the spatial scan statistic, we can propose a new spatial scan statistic by considering the following \textit{regression model without the intercept term $\alpha$} under the Gaussian model. 
\begin{thm}\label{prop expectation based full para gauss SSS}
We consider a regression model without the intercept term $\alpha$ under the Gaussian model. 
\begin{align*}
f (y_{i} | Z_{k,i}; 0, \theta_{S_k} , \sigma^2  ) = \frac{1}{ \sqrt{2 \pi \sigma^2 } } \exp \left\{ - \frac{ ( y_i- \mu_i )^2 }{ 2 \sigma^2 } \right\}, 
\end{align*}
where $ \mu_i  =  \theta_{S_k} Z_{k,i} $. 
Then, we have the following new spatial scan statistic: 
\begin{align*}
\mathrm{LLR}_k =   \frac{N}{2} \log \left( \frac{1}{N} \sum_{i=1}^N y_i^2 \right) - \frac{N}{2} \log   \frac{1}{N} \left\{  \sum_{i=1}^N y_i^2 -  \frac{ \left( \sum_{i=1}^N y_i   Z_{k,i} \right)^2  }{ \sum_{i=1}^N Z_{k,i} }  \right\}   .
\end{align*}
\end{thm}
The proof is in the Appendix. 
By a similar analogy to the case of the Poisson model, the spatial scan statistic in Proposition~\ref{prop population based full para gauss SSS} and Theorem~\ref{prop expectation based full para gauss SSS} can be regarded as the population-based and expectation-based spatial scan statistic, respectively. 

\subsubsection{Case of Bernoulli Model}
For the Bernoulli model, \cite{kulldorff1997} has proposed the classical spatial scan statistic, and \cite{jung2009} re-examined it by considering the regression-based spatial scan statistic. 
\cite{jung2009} considered the following regression model:
\begin{align*}
f (y_{i} | Z_{k,i}; \alpha, \theta_{S_k}  ) =\pi_i^{ y_{i} } (1- \pi_i)^{( 1 - y_i )}  , 
\end{align*}
where $  \log \frac{ \pi_i }{ 1 - \pi_i }  = \alpha + \theta_{S_k} Z_{k,i}  $. 
The corresponding spatial scan statistic is as follows: 
\begin{align}\label{bernoulli original SSS}
& \mathrm{LLR}_k  \nonumber \\
& = \left\{ \left( \sum_{i=1}^N y_i Z_{k,i} \right) \log \frac{ \sum_{i=1}^N y_i Z_{k,i} }{ \sum_{i=1}^N Z_{k,i} 
 } + \left( \sum_{i=1}^N (1 - y_i) Z_{k,i} \right) \log \frac{ \sum_{i=1}^N (1 - y_i) Z_{k,i} }{ \sum_{i=1}^N Z_{k,i} } \right. \nonumber \\
 &  \left.  + \left( \sum_{i=1}^N (1 - Z_{k,i} ) y_i \right) \log \frac{\sum_{i=1}^N (1 - Z_{k,i} ) y_i 
 }{  \sum_{i=1}^N (1 - Z_{k,i} ) } + \left( \sum_{i=1}^N (1 - Z_{k,i} ) ( 1 - y_i ) \right) \log \frac{  \sum_{i=1}^N (1 - Z_{k,i} ) ( 1 - y_i ) }{ \sum_{i=1}^N (1 - Z_{k,i} )  } \right\} \nonumber \\
 & - \left\{  \left( \sum_{i=1}^N y_i \right) \log \frac{ \sum_{i=1}^N y_i }{ N } + \left( \sum_{i=1}^N (1 -  y_i ) \right) \log \frac{ \sum_{i=1}^N (1 - y_i ) }{ N } \right\}.
\end{align}
Similarly to previous examples in the Poisson and Gaussian models, we propose a new spatial scan statistic considering the \textit{regression model without the intercept term $\alpha$}. 
\begin{thm}\label{new exp. bernoulli} 
We consider the following regression model without the intercept term $\alpha$ under the Bernoulli model: 
\begin{align*}
f (y_{i} | Z_{k,i}; 0, \theta_{S_k}  ) =\pi_i^{ y_{i} } (1- \pi_i)^{( 1 - y_i )}  , 
\end{align*}
where $  \log \frac{ \pi_i }{ 1 - \pi_i }  =  \theta_{S_k} Z_{k,i}  $. 
Then, the log-likelihood ratio~\eqref{llr def} yields the new spatial scan statistic. 
\begin{align*}
 \mathrm{LLR}_k & = \left( \sum_{i=1}^N y_i Z_{k,i}  \right) \log \frac{ \sum_{i=1}^N y_i Z_{k,i} }{ \sum_{i=1}^N Z_{k,i} } + \left( \sum_{i=1}^N ( 1 - y_i ) Z_{k,i} \right) \log \frac{ \sum_{i=1}^N ( 1 - y_i ) Z_{k,i} }{ \sum_{i=1}^N Z_{k,i} } \\
& \qquad  +  \left( \sum_{i=1}^N Z_{k,i} \right) \log 2 .
\end{align*}
\end{thm}
The proof is in the Appendix. 
%
%\begin{rem}
By a similar analogy to the case of the Poisson and Gaussian models, the spatial scan statistic in~\eqref{bernoulli original SSS} and Theorem~\ref{new exp. bernoulli} can be regarded as the population-based and expectation-based one, respectively. 
%\end{rem}
In addition, the difference in the mean structure of the regression model can be interpreted as follows: 
\begin{align*}
\frac{ \pi_i }{ 1- \pi_i } = \exp( \alpha ) \times \exp(\theta_{ S_k } Z_{k,i} ).
\end{align*}
Under no clusters, i.e., $ \theta_{S_k} = 0 $, the baseline value of the odds ratio is $\exp(\alpha)$ and $\exp(0) = 1$ for the population-based and expectation-based spatial scan statistics, respectively.

%%%Sect. 3 %%%
\section{Spatial Scan Statistics as Penalized Regression Model}
As shown in Sect.~\ref{subsect1 sp scan brief}, the main procedure of spatial scan methods is to find the MLC among potential clusters. 
Recall that the MLC is a potential cluster that attains the maximum value of the log-likelihood ratio~\eqref{llr def}, given by 
\begin{align*}
\mathrm{LLR}_k & =  \log \frac{ \max_{ (\alpha, \theta_{S_{k}}, \bs{\beta}^{\top}   ) \in \Theta }  \Pi_{i=1}^N  f (y_{i} | x_{i}, Z_{k,i}; \alpha, \theta_{S_{k}}, \bs{ \beta }^{\top}  ) }{   \max_{ (\alpha, \bs{\beta}^{\top}  ) \in \Theta_{0} }  \Pi_{i=1}^N   
f (y_{i} | x_{i}, Z_{k,i}; \alpha, 0, \bs{ \beta }^{\top}  ) } \\
& = \log  \max_{ (\alpha, \theta_{S_{k}}, \bs{\beta}^{\top}   ) \in \Theta }  \Pi_{i=1}^N  f (y_{i} | x_{i}, Z_{k,i}; \alpha, \theta_{S_{k}}, \bs{ \beta }^{\top}  ) \\
& \quad - \log  \max_{ (\alpha, \bs{\beta}^{\top}  ) \in \Theta_{0} }  \Pi_{i=1}^N f (y_{i} | x_{i}, Z_{k,i}; \alpha, 0, \bs{ \beta }^{\top}  ) \\
& = \max_{ (\alpha, \theta_{S_{k}}, \bs{\beta}^{\top}   ) \in \Theta } \sum_{i=1}^N \log f (y_{i} | x_{i}, Z_{k,i}; \alpha, \theta_{S_{k}}, \bs{ \beta }^{\top}  ) \\
& \quad - \max_{ (\alpha, \bs{\beta}^{\top}  ) \in \Theta_{0} } \sum_{i=1}^N \log f (y_{i} | x_{i}, Z_{k,i}; \alpha, 0, \bs{ \beta }^{\top}  ).
\end{align*}
The second term does not depend on $S_k$. 
Thus, \eqref{llr max}, which determines the index of the MLC, is rewritten as follows:
\begin{align}\label{modified llr max}
 \argmax_{ k \in \left\{ 0, \ldots, K \right\} } \left\{ \max_{ (\alpha, \theta_{S_{k}}, \bs{\beta}^{\top} ) \in \Theta }  \sum_{i=1}^N \log f (y_{i} | \mb{x}_{i}, Z_{k,i}; \alpha, \theta_{S_{k}}, \bs{ \beta }^{\top} ) \right\},
\end{align}
where $ \theta_{S_k} =0 $ and $  Z_{k,i} \theta_{S_k} = 0 $ for $k=0$. 

\subsection{Equivalence to Regression with $\ell_0$ Penalty}\label{sect equiv to reg with l0}
We consider the following regression model to show the connection between \eqref{modified llr max} and a penalized regression model with the well-known sparse $\ell_0$ penalty. 
\begin{align}\label{extended reg model}
g( \mathrm{E}_{f (y_{i} | \mb{x}_{i},  \left\{ Z_{k,i} \right\}_{k=1}^{K}  ; \mu_{i} )} [ y_{i} ] )= g ( \mu_{i} ) = \alpha + \sum_{k=1}^K \theta_{S_{k}} Z_{k ,i} +  \mb{x}_{i}^{\top} \bs{\beta}.
\end{align}
This formulation~\eqref{extended reg model} appears multiple times in the topic of finding multiple clusters~\citep{takahashi2018,takahashi2020}, such as how many clusters to select in the context of spatial scan statistics, and spatial regression methods for the cluster detection~\citep{XU201659,KAMENETSKY2022100462}. 
However, there are no studies on the relationship between obtaining the MLC, the main procedure for spatial scan statistics. 
Due to this regression model,~\eqref{modified llr max} can be viewed as selecting one potential cluster among $ S_{0}, S_{1}, \ldots , S_{K} $ via~\eqref{extended reg model} as follows: 
\begin{thm}\label{thm reg with l0 pen}
By \eqref{extended reg model} and the $\ell_0$ penalty, \eqref{modified llr max} is equivalent to 
\begin{align}\label{opt problem llr}
\argmax_{ (\alpha, \bs{\theta}^{\top} , \bs{\beta}^{\top} ) } \frac{1}{N} \sum_{i=1}^N \log f (y_{i} | \mb{x}_{i}, \left\{ Z_{k,i} \right\}_{k=1}^{K} ; \alpha, \bs{\theta}^{\top} , \bs{\beta}^{\top}) \ \mbox{ s.t. }  \| \bs{\theta}  \|_0 \leq 1,
\end{align}
where $\bs{\theta} =(\theta_{S_{1}} , \ldots , \theta_{S_{K}})^{\top} $ and $ \| \bs{\theta} \|_0 = \sum_{k=1}^K \mathrm{I} ( \theta_{ S_{k} } \neq 0) $, where $\mathrm{I}$ is the indicator function.
\end{thm}
\begin{proof}
Feasible solutions under the constraint $\| \bs{\theta} \|_{0} \leq 1 $ imply that all elements of $\bs{ \theta}$ are zeros, or one element of $\bs{ \theta }$ is non-zero. 
The former indicates that the MLC is $S_0$, i.e., there is no cluster among potential clusters. 
The latter indicates that the MLC is any of $ \left\{ S_1, \ldots, S_K \right\} $ with the largest log-likelihood value. 
\end{proof}
Note that, for finding the MLC, the original definition of MLC~\eqref{llr max} is equivalent to~\eqref{opt problem llr}. 
This theorem does not imply that it includes obtaining the p-value to evaluate the statistical significance of the obtained MLC via the Monte Carlo hypothesis testing, which has been used in ordinal spatial scan statistics. 
By reformulating the optimization problem with the sparse penalty~\eqref{opt problem llr}, we can apply another method of statistical inference for calculating the p-value, such as the post-selection inference~\citep{lee2016exact,tibshirani2016exact} and the bootstrap-based method~\citep{10.1214/14-AOS1221}. 
From a theoretical perspective, the consistency and asymptotic normality of the estimator for spatial scan statistics can be investigated in the context of sparse regression.

In addition, we can easily show a relation to existing spatial scan methods for finding \textit{non-overlapping} multiple clusters~\citep{Zhang2010SpatialSS,li2011}.
\begin{cor}\label{corr multiple zones}
Under given non-overlapping potential clusters, i.e., $S_{l} \cap S_{m} = \emptyset  \mbox{ for any } 1 \leq l, m \leq K$, one can select multiple zones by changing the strength of the penalty in \eqref{opt problem llr}, i.e., $  \| \bs{\theta} \|_0 \leq c$, where $c \geq 2$. 
\end{cor}

We consider the following regression model \eqref{extended reg model} with time index $t$ for spatio-temporal data. 
\begin{align*}
g( \mathrm{E}_{f (y_{t,i} | \mb{x}_{t,i},  \left\{ Z_{t,k,i} \right\}_{k=1}^K   ; \mu_{t,i} )} [ y_{t,i} ] )= g ( \mu_{t,i} ) = \alpha_{t} + \sum_{k=1}^K \theta_{S_{t,k}} Z_{t,k,i} +  \mb{x}_{t,i}^{\top} \bs{\beta}_{t}.
\end{align*}
Then, \eqref{opt problem llr} is easily extended to a space-time scan method as follows: 
\begin{align*}
\argmax_{ ( \bs{\alpha}^{\top}, \bs{\theta}^{\top} , \bs{\beta}^{\top} ) } \frac{1}{N \times T} \sum_{t=1}^T \sum_{i=1}^N \log f (y_{t,i} | \mb{x}_{t,i},  \left\{ Z_{t,k,i} \right\}_{k=1}^K ; \bs{\alpha}^{\top}, \bs{\theta}^{\top} , \bs{\beta}^{\top}) \ \mbox{ s.t. }  \| \bs{\theta}  \|_0 \leq 1,
\end{align*}
where $T$ is the time period, $ \bs{\alpha} = ( \alpha_{1} , \ldots , \alpha_{T} )^{\top}  $,  $\bs{\theta} =(\theta_{S_{1,1}} , \ldots, \theta_{S_{1,K}} ,  \ldots , \theta_{ S_{T,1} } , \ldots , \theta_{S_{T,K}})^{\top} $, and $\bs{\beta} = ( \bs{\beta}_{1}^{\top} , \ldots, \bs{\beta}_{T}^{\top}  )^{\top} $.
By setting some parts of $\bs{\alpha}$ and $\bs{\theta}$ to be the same, e.g., $\alpha_1 = \alpha_2 = \alpha_3 $ and $\theta_{S_{1,1}} = \theta_{S_{2,1}} = \theta_{S_{3,1}}$, respectively, our formulation without covariates corresponds to the existing space-time scan statistic such as the cylindrical space-time scan statistic~\citep{kull2001}.

\subsection{Connection to Existing Spatial Regression Method for Cluster Detection}\label{sect related work}
Here, we show existing spatial cluster detection methods related to the regression model~\eqref{extended reg model} and optimization problem~\eqref{opt problem llr}. 

\cite{XU201659} adopted forward step-wise and stage-wise approaches with several stopping criteria such as AIC instead of directly solving \eqref{opt problem llr}. 
\cite{KAMENETSKY2022100462} considered the Lagrange relaxation form replaced by $\ell_1$ penalty instead of the $\ell_0$ penalty in \eqref{opt problem llr} as follows: 
\begin{align}\label{kame proposed} 
\argmax_{ (\alpha, \bs{\theta}^{\top} , \bs{\beta}^{\top} ) } \frac{1}{N} \sum_{i=1}^N \log f (y_{i} | \mb{x}_{i},  \left\{ Z_{k,i} \right\}_{k=1}^K  ; \alpha, \bs{\theta}^{\top} , \bs{\beta}^{\top}) - \lambda  \| \bs{\theta}  \|_1 ,
\end{align}
where $\lambda$ is a tuning parameter and $ \| \bs{\theta} \|_1 = \sum_{k=1}^K | \theta_{ S_{k} } | $. 
To solve \eqref{kame proposed}, \cite{KAMENETSKY2022100462} adopted a general optimization approach such as a coordinate descent~\citep{JSSv033i01}. 

\cite{wang2014} considered the following regression model. 
\begin{align}\label{wang reg model}
g( \mathrm{E}_{f (y_{i} | x_{i} ; \mu_{i} )} [ y_{i} ] )= g ( \mu_{i} ) = \alpha +  \theta_{i} +  \mb{x}_{i}^{\top} \bs{\beta}. 
\end{align}
It can be regarded as the spatial scan under the case of potential clusters, which consist of each region, i.e., $S_1 = \left\{ 1 \right\}, S_2 = \left\{ 2 \right\}, \cdots, S_N = \left\{ N \right\} $ in \eqref{extended reg model}. 
To obtain the MLC, they consider the following optimization problem. 
\begin{align}\label{wang proposed opt}
\argmax_{ (\alpha, \theta_{1}, \ldots, \theta_{N}, \bs{\beta}^{\top} ) } \frac{1}{N} \sum_{i=1}^N \log f (y_{i} | \mb{x}_{i}; \theta_i , \bs{\beta}^{\top} ) - \lambda \| \bs{\theta}  \|_1  - \lambda^{\prime} \sum_{(j,l ) \in \mathcal{E}} | \theta_{j} - \theta_{l} | ,
\end{align}
where $\lambda^{\prime}$ is a tuning parameter for graph-fused lasso~\citep{tib2005}, and $\mathcal{E}$ is the set of edges in undirected graph $\mathcal{G}$, which consists of adjacency matrix of entities $(i=1, \ldots, N)$. %
\cite{choi2018} developed an estimation algorithm to treat with  $\ell_1$ penalty of coefficients for covariates, i.e.,  $\| \bs{\beta} \|_1$, in \eqref{wang proposed opt} by using the idea of Majorization-Minimization algorithm~\citep{hunter2004}. 
Regarding statistical inference, it is generally challenging to construct confidence intervals in penalized regression models. 
To overcome such a problem, a Bayesian formulation of \eqref{wang proposed opt} was proposed by \cite{ryo2022}.

\subsection{Spatial Scan Statistics Meet Robust Statistics}\label{sect spatial scan robust}
We consider the potential clusters consisting of each region, i.e., $S_1 = \left\{ 1 \right\}, S_2 = \left\{ 2 \right\}, \cdots, S_N = \left\{ N \right\} $. 
Recall that the following regression model \eqref{wang reg model}:
\begin{align*}
g( \mathrm{E}_{f (y_{i} | x_{i} ; \mu_{i} )} [ y_{i} ] )= g ( \mu_{i} ) = \alpha +  \theta_{i} +  \mb{x}_{i}^{\top} \bs{\beta}. 
\end{align*}
For simplicity, we consider the case of the Gaussian model with fixed variance $ \sigma^2 $, i.e., 
\begin{align*}
f(y_{i} | x_{i} ; \mu_{i} ) =  \frac{1}{ \sqrt{2 \pi \sigma^2 } } \exp \left\{ - \frac{ ( y_i -  \alpha -  \theta_{i} -  \mb{x}_{i}^{\top} \bs{\beta}  )^2   }{ 2 \sigma^2 } \right\} .
\end{align*}
Here we focus on the parameter $\mb{\theta} = (\theta_1, \ldots, \theta_N )^{\top} $. 
Under these settings, the Lagrange relaxation form of the optimization problem for obtaining the MLC is written as follows:
\begin{align}\label{opt mean shift para}
\argmax_{ (\theta_1, \ldots, \theta_N) } - \frac{1}{N} \sum_{i=1}^{N} \left(  \tilde{y}_i -  \theta_{i}   \right)^2  - \lambda \| \mb{\theta} \|_0 + \mbox{others},
\end{align}
where $\tilde{y}_i =  y_i - \alpha - \mb{x}_{i}^{\top} \bs{\beta} $ and others does not depend on the parameter $\mb{\theta}$. 
As pointed out in~\cite{YGchoi2022}, a parameter $\theta_i$ seems to be a case-specific parameter investigated by~\cite{she2011} and~\cite{10.1214/11-STS377} as absorbing harmful effects of outlying $y_i$.
We explain intuitively why the above formulation becomes robust against outliers. 
Let $(\alpha^*, {\bs{\beta}^* }^{\top} )$ be true regression coefficients. 
Outlying $y_i$ means $ | y_i - \alpha^* - \mb{x}_{i}^{\top} \bs{\beta}^* | \gg 0 $, then, the mean shift parameter $\theta_i$ can reduce it, and $ \| \mb{\theta} \|_0 $ controls the number of outlying $y$. 
Therefore, the baseline $\alpha + \mb{x}_{i}^{\top} \bs{\beta}$ can be correctly estimated, even under the presence of outliers. 

The more rigorous explanation related to robust statistics is as follows. 
By following the formulation in~\cite{she2011}, the above optimization problem~\eqref{opt mean shift para} yields in 
\begin{align*}
\argmin_{ (\alpha, \mb{\beta}^{\top} ) }  \frac{1}{N} \sum_{i=1}^{N} \Psi \left(  y_i - \alpha - \mb{x}_{i}^{\top} \bs{\beta} ; \lambda \right)  ,
\end{align*}
where $\Psi (a;\lambda)  = a^2 \mbox{ if } |a| \leq \lambda \mbox{ and } = 0 \mbox{ otherwise} $.
The function $\Psi$ is known as skipped-mean loss and has robustness against outliers. 
It has been shown that such a relationship to robustness holds for another penalty, e.g., $\ell_1$ penalty corresponds to the Huber loss (see,~\cite{she2011} and~\cite{10.1214/11-STS377} for details). 
In addition, one can easily extend the above relation to other regression models, such as GLMs~\citep{10.1214/11-STS377} and non-overlapping potential clusters, i.e., $S_{l} \cap S_{m} = \emptyset $. 
Returning to spatial scan statistics, $\theta_i \neq 0$ indicates the possibility of being a cluster, and it is an outlier in terms of robust statistics. 
Therefore, most spatial scanning methods detect outliers, and any of them can be the MLC. 
%

%%%Sect. 4 %%%
\section{Application to Real Data}\label{sect appl real data}
We conducted numerical experiments using the New York Leukemia dataset, which is available in R package \texttt{smerc}. 
The dataset contains 281 observations from an 8-county area in the state of New York and its geographical information served as the spatial framework for our simulations. 
Since true disease clusters are typically unknown in real-world scenarios, we artificially create cluster scenarios (hot and cold spots) for the purpose of validation. 
Figure 1 shows a map of our cluster scenarios. 
\begin{figure}[!ht]
    \centering
    \includegraphics[width=0.4\textwidth]{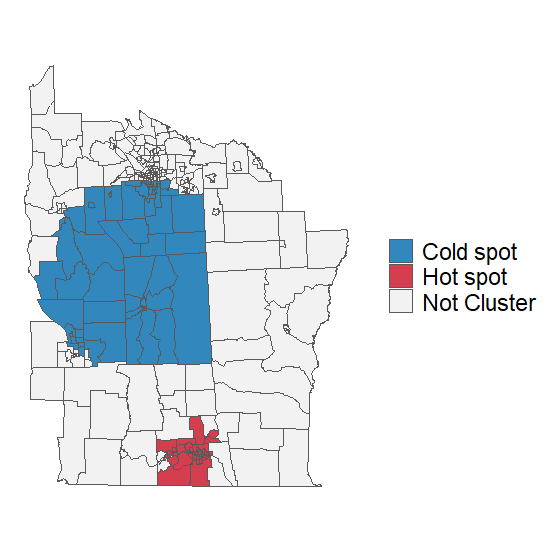} %
    \caption{Cluster scenarios for hot and cold spots} 
    \label{fig:1} 
\end{figure}
The number of observations is 40 and 61 for the hot and cold spots, respectively. 
Disease occurrence counts were randomly generated based on the normal distribution, with two distinct data generation schemes: 
\begin{description}
\item[Population-based setting:] $ y_i \sim \mathrm{N} ( \alpha_{pop} + \theta_{ pop } Z_{k,i} , \sigma_{ pop }) $,
\item[Expectation-based setting:] $ y_i \sim \mathrm{N} (  \theta_{ expe } Z_{k,i} , \sigma_{ expe }) $,
\end{description}
where $\mathrm{N}$ is the normal distribution, $\alpha_{pop} = 5, \theta_{pop} = \theta_{expe} = 5 \mbox{ for hot spot }, -5 \mbox{ for cold spot}, \mbox{ and } \sigma_{pop} = \sigma_{expe} = 0.5$.
Our comparative analysis focused on two spatial scan statistics: The newly proposed spatial scan statistic in Theorem~\ref{prop expectation based full para gauss SSS} and Kulldorff's spatial scan statistic (2009) in Proposition~\ref{prop population based full para gauss SSS} (Kull2009). 
As discussed in Sect.~\ref{connect ex scan gauss model}, the former can be considered an expectation-based spatial scan statistic, while the latter is a population-based spatial scan statistic. 
Potential clusters were constructed using the circular-based method, and the number of potential clusters was $22548$. 
For each simulation setting, we calculated log-likelihood ratio values for all potential clusters and identified the top three potential clusters: the MLC, the second likely cluster, and the third likely cluster. 
In addition, for the Kull2009, we showed the estimates of the intercept term $\hat{\alpha} = \frac{ \sum_{i=1}^{N} (1-Z_{k,i}) y_i  }{ N - \sum_{i=1}^{N} Z_{k,i} } $ in the numerator of the log-likelihood ratio for all potential clusters $S_k$ in the expectation-based setting. 
This represents a deviation from the true value of the intercept term, which should be zero in an expectation-based setting and is likely to affect statistical inferences, such as when calculating p-values in Monte Carlo hypothesis tests. 
The primary objective of this experiment is to investigate their behavior under the specified data scenarios and not to determine which of the two methods is superior. 

Figures 2 and 3 visualize these results, highlighting whether the artificially defined ground truth clusters were detected. 
\begin{figure}[!ht]
    \centering
    \includegraphics[width=\textwidth]{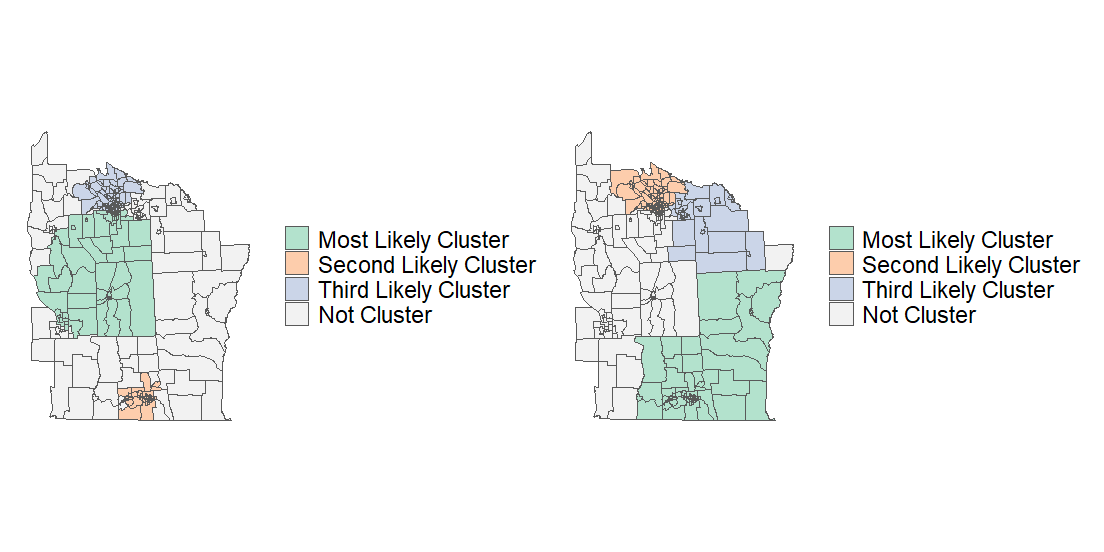} %
    \caption{Results for population-based setting: Kull2009 (Left) and Proposed (Right)} 
    \label{fig:numericalexpe1} 
\end{figure}
\begin{figure}[!ht]
    \centering
    \includegraphics[width=\textwidth]{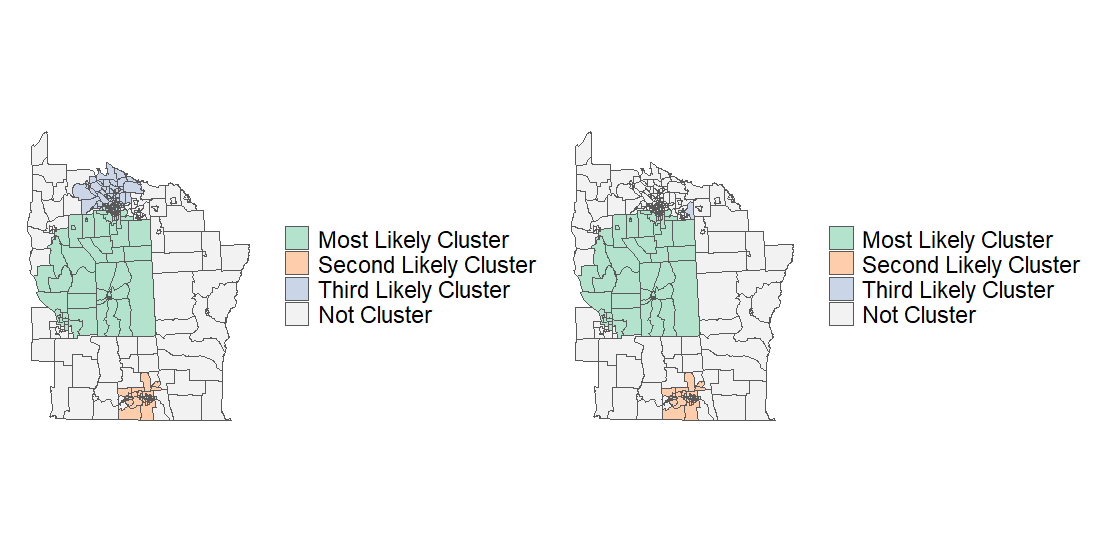} % 
    \caption{Results for expectation-based setting: Kull2009 (Left) and Proposed (Right)}
    \label{fig:numericalexpe2} % 
\end{figure}
For the population-based setting in Figure 2, our proposed spatial scan statistic can detect the hot spot as part of the MLC, i.e., including irrelevant regions, but fails to detect the cold spot. 
On the other hand, the Kull2009 can detect them correctly. 
For the expectation-based setting in Figure 3, both methods can detect true clusters. 
However, Figure 4 shows that the histogram of the estimated value of the intercept term for the Kull2009 has shifted from its true value of zero. 
\begin{figure}[!ht]
    \centering
    \includegraphics[width=0.5\textwidth]{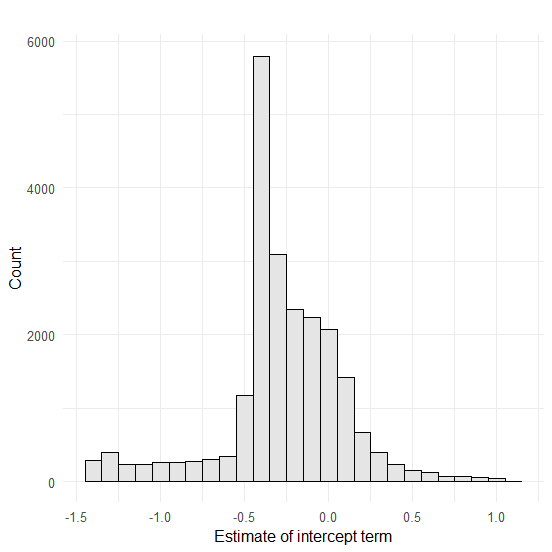} %
    \caption{Histgram of estimates of the intercept term for Kull2009 under expectation-based setting} 
    \label{fig:hist} 
\end{figure}
It also implies that the estimation of $\theta$, which is of interest for the statistical inference, is affected. 
It indicates that one should select the appropriate selection of spatial scan statistics for the type of data that are analyzed, as stated in~\cite{neillphdthesis}. 
%

%%%Sect.5 %%%
\section{Conclusion}\label{sect conclusion} 
By analyzing regression-based spatial scan statistics, we clarified the theoretical distinctions between population-based and expectation-based approaches, which had only been partially explored. 
The key difference lies in the presence or absence of an intercept term in the regression model. 
Building on these findings, we developed new expectation-based spatial scan statistics for models such as the normal and Bernoulli distributions. 
Additionally, we formulated an equivalent optimization problem for obtaining the MLC by introducing an extended regression model with an $\ell_0$ penalty. 
The proposed framework can be readily extended to space-time scan statistics and methods for detecting multiple clusters. 
By framing the problem as a penalized regression, we found that the proposed approach is closely related to various existing advanced spatial regression-based cluster detection methods. 
In the context of robust statistics for handling outliers, it was also revealed that clusters in spatial scan statistics correspond to outliers in robust regression methods. 
Numerical experiments on real-world data demonstrated different behaviors of two spatial scan statistics, including the proposed one under specified data scenarios. 
It again suggests the importance of selecting spatial scan statistics appropriate for the nature of the data, as already stated in~\cite{neillphdthesis}. 
However, certain properties, such as those stated in Corollary~\ref{corr multiple zones}, remain unverified in cases involving overlapping potential clusters.
Addressing these limitations is a subject for future research.

\section*{Acknowledgments}
This work was partially supported by a KAKENHI Grant-in-Aid for Young Scientists (22K17859),
%%%

%
\bibliography{biomtemplate.bib}

\section*{Appendix}
\subsection*{Proof of Proposition~\ref{prop equiv exp based under poisson}}
%\begin{proof}
%
\begin{align*}
& \mathrm{LLR}_k \\
& = \log \frac{ \max_{  ( 0, \theta_{S_{k}} )  }  \Pi_{i=1}^N  f (y_{i} | Z_{k,i};  0, \theta_{S_{k}} ) }{  \Pi_{i=1}^N   
f (y_{i} |  Z_{k,i};   0, 0 ) } \\
& = \max_{ ( 0, \theta_{S_k} ) } \left\{  \sum_{i=1}^N \log f (y_{i} | Z_{k,i};  \theta_{S_{k}} )   \right\} - \sum_{i=1}^N \log f (y_{i} | Z_{k,i}; 0,  0 ) \\
& = \max_{ ( 0, \theta_{S_{k}} ) } \left[ \sum_{i=1}^N \left\{  y_i \left(  \theta_{S_k} Z_{k,i} + \log ( \gamma_i ) \right) - \exp ( \theta_{S_k} Z_{k,i} + \log( \gamma_i ) )     \right\} \right] - \sum_{i=1}^N \left\{  y_i  \log ( \gamma_i ) - \gamma_i     \right\} \\
& = \max_{ ( 0, \theta_{S_{k}} ) } \left[ \sum_{i=1}^N \left\{   y_i \theta_{S_k} Z_{k,i}  - \gamma_i \exp ( \theta_{S_k} Z_{k,i} )   \right\} \right] + \sum_{i=1}^N \gamma_i  \\
& \mbox{The first term attains the maximum when } \hat{\theta}_{S_k} = \log \left( \frac{ \sum_{i=1}^N y_i Z_{k,i} }{  \sum_{i=1}^N \gamma_i Z_{k,i} } \right) . \\
& = \left( \sum_{i=1}^N y_i Z_{k,i} \right) \log \left( \frac{\sum_{i=1}^N  y_i Z_{k,i} }{\sum_{i=1}^N  \gamma_i Z_{k,i} } \right) - \sum_{i=1}^N  \gamma_i   \left( \frac{ \sum_{i=1}^N y_i Z_{k,i} }{  \sum_{i=1}^N \gamma_i Z_{k,i} }  \right)^{ Z_{k,i} }   + \sum_{i=1}^N \gamma_i \\
& = \left( \sum_{i=1}^N y_i Z_{k,i} \right) \log \left( \frac{\sum_{i=1}^N  y_i Z_{k,i} }{\sum_{i=1}^N  \gamma_i Z_{k,i} } \right) - \sum_{i=1}^N  y_i Z_{k,i} - \sum_{i=1}^N (1 - Z_{k,i}) \gamma_i   + \sum_{i=1}^N \gamma_i \\
& = \left( \sum_{i=1}^N y_i Z_{k,i} \right) \log \left( \frac{ \sum_{i=1}^N y_i Z_{k,i} }{  \sum_{i=1}^N \gamma_i Z_{k,i} } \right) + \sum_{i=1}^N ( \gamma_i - y_i )Z_{k,i} .
\end{align*}
%\end{proof}

%
\subsection*{Proof of Proposition~\ref{prop equiv exp based under gauss} for population-based and expectation-based}
%
%\begin{proof}
For population-based, we have the following log-likelihood ratio
\begin{align*}
& \mathrm{LLR}_k \\ 
& = \log \frac{ \max_{ ( \alpha, \theta_{S_{k}} ) }  \Pi_{i=1}^N  f (y_{i} | Z_{k,i};  \alpha, \theta_{S_{k}} ) }{  \max_{ ( \alpha, 0 )}  \Pi_{i=1}^N   
f (y_{i} |  Z_{k,i}; \alpha, 0 ) } \\ 
& = \max_{ ( \alpha, \theta_{S_k} ) } \left\{  \sum_{i=1}^N \log f (y_{i} | Z_{k,i}; \alpha, \theta_{S_{k}} )   \right\} - \max_{ ( \alpha, 0 ) } \left\{  \sum_{i=1}^N \log f (y_{i} | Z_{k,i}; \alpha, 0 ) \right\} \\
& = \max_{ ( \alpha, \theta_{S_k} ) }  \left\{ \sum_{i=1}^N  \log \frac{1}{ \sqrt{2 \pi \sigma_i^2} }  - \frac{ (y_i - \gamma_i( 1 +\alpha + \theta_{S_k} Z_{ k,i }) )^2 }{ 2 \sigma_i^2 }   \right\} \\ 
& \quad - \max_{ ( \alpha, 0 ) } \left\{ \sum_{i=1}^N  \log \frac{1}{ \sqrt{2 \pi \sigma_i^2} } - \frac{ (y_i - \gamma_i ( 1 + \alpha ) )^2 }{ 2 \sigma_i^2 }   \right\} \\
& \mbox{The first term attains the maximum when } \\
&  \hat{\theta}_{ S_k } =  \frac{ \sum_{i=1}^N \frac{ y_i \gamma_i }{ \sigma_i^2 } Z_{k,i} }{ \sum_{i=1}^N \frac{ \gamma_i^2 }{ \sigma_i^2 }  Z_{k,i} } -  \frac{ \sum_{i=1}^N \frac{ y_i \gamma_i }{ \sigma_i^2 } ( 1- Z_{k,i} ) }{ \sum_{i=1}^N \frac{ \gamma_i^2 }{ \sigma_i^2 } ( 1- Z_{k,i} ) } \mbox{ and } \hat{ \alpha } = \frac{ \sum_{i=1}^N \frac{ y_i \gamma_i }{ \sigma_i^2 } ( 1- Z_{k,i} ) }{ \sum_{i=1}^N \frac{ \gamma_i^2 }{ \sigma_i^2 } ( 1- Z_{k,i} ) } -1 , \\
& \mbox{and the second term does when } \\
& \bar{ \alpha }= \frac{ \sum_{i=1}^N \frac{ y_i \gamma_i }{ \sigma_i^2 } }{ \sum_{i=1}^N \frac{ \gamma_i^2 }{ \sigma_i^2 } } -1. \\
& = \sum_{i=1}^N   - \frac{ (y_i Z_{k,i} - \gamma_i( \hat{ \alpha } + \hat{ \theta }_{S_k} ) Z_{ k,i } + y_i ( 1 - Z_{k,i} ) - \gamma_i \hat{ \alpha } (1 - Z_{k,i}  )  )^2 }{ 2 \sigma_i^2 } \\
& \quad - \sum_{i=1}^N - \frac{ (y_i - \gamma_i \bar{ \alpha } )^2 }{ 2 \sigma_i^2 }  \\
& = \sum_{i=1}^N   - \frac{ (y_i  - \gamma_i( \hat{ \alpha } + \hat{ \theta }_{S_k}  ) )^2 }{ 2 \sigma_i^2 } Z_{k,i} + \sum_{i=1}^N - \frac{ ( y_i  - \gamma_i \hat{ \alpha } )^2 }{ 2 \sigma_i^2 } (1 - Z_{k,i} )  \\  
& \quad + \sum_{i=1}^N \frac{ y_i^2 }{2 \sigma_i^2} -  \frac{ \left( \sum_{i=1}^N \frac{ y_i \gamma_i }{ \sigma_i^2 } \right)^2 }{2 \sum_{i=1}^N \frac{ \gamma_i^2 }{\sigma_i^2 } } \\
& = \frac{ \left( \sum_{i=1}^N \frac{y_i \gamma_i }{ \sigma_i^2} Z_{k,i} \right)^2 }{ 2 \sum_{i=1}^N \frac{ \gamma_i^2 }{ \sigma_{i}^2 } Z_{k,i}  } + \frac{ \left( \sum_{i=1}^N \frac{y_i \gamma_i }{ \sigma_i^2} ( 1 - Z_{k,i} ) \right)^2 }{ 2 \sum_{i=1}^N \frac{ \gamma_i^2 }{ \sigma_{i}^2 } ( 1 - Z_{k,i} )  } -\frac{ \left( \sum_{i=1}^N \frac{y_i \gamma_i }{ \sigma_i^2}  \right)^2 }{ 2 \sum_{i=1}^N \frac{ \gamma_i^2 }{ \sigma_{i}^2 }  }.
\end{align*}
%\end{proof}
%
%
%\begin{proof}

For expectation-based, we have the following log-likelihood ratio
\begin{align*}
& \mathrm{LLR}_k \\ 
& = \log \frac{ \max_{ (0, \theta_{S_{k}} )  }  \Pi_{i=1}^N  f (y_{i} | Z_{k,i}; 0, \theta_{S_{k}} ) }{    \Pi_{i=1}^N   
f (y_{i} |  Z_{k,i};  0 , 0 ) } \\ 
& = \max_{ ( 0, \theta_{S_k} ) } \left\{  \sum_{i=1}^N \log f (y_{i} | Z_{k,i};  0, \theta_{S_{k}} )   \right\} - \sum_{i=1}^N \log f (y_{i} | Z_{k,i};  0, 0 ) \\
& = \max_{ ( 0, \theta_{S_k} ) }  \left[ \sum_{i=1}^N  \log \frac{1}{ \sqrt{2 \pi \sigma_i^2} }  - \frac{ (y_i - \gamma_i - \gamma_i \theta_{S_k} Z_{ k,i } )^2 }{ 2 \sigma_i^2 }   \right] - \sum_{i=1}^N \left\{  \log \frac{1}{ \sqrt{2 \pi \sigma_i^2} } - \frac{ (y_i - \gamma_i)^2 }{ 2 \sigma_i^2 }   \right\} \\
& = \max_{ ( 0, \theta_{S_k} ) } \left\{ \sum_{i=1}^N  - \frac{  (y_i - \gamma_i - \gamma_i \theta_{S_k} Z_{ k,i } )^2 }{ 2 \sigma_i^2 } \right\} + \sum_{i=1}^N \frac{ (y_i - \gamma_i)^2 }{ 2 \sigma_i^2 }     \\
& \mbox{The first term attains the maximum when } \hat{\theta}_{S_k} =  \frac{ \sum_{i=1}^N \frac{ y_i \gamma_i }{ \sigma_i^2 }  Z_{k,i} }{  \sum_{i=1}^N \frac{ \gamma_i^2 }{ \sigma_i^2 } Z_{k,i} } -1  .  \\
& =   \frac{ 1- ( 1 + \hat{\theta}_{S_k} )^2 }{2} \sum_{i=1}^N \frac{ \gamma_i^2 }{ \sigma_i^2 } Z_{k,i} +  \hat{\theta}_{S_k}  \sum_{i=1}^N \frac{ y_i \gamma_i }{ \sigma_i^2 } Z_{k,i} \\
& = \frac{ \left( \sum_{i=1}^N \frac{ y_i \gamma_i }{ \sigma_i^2 }  Z_{k,i} \right)^2 }{ 2 \sum_{i=1}^N \frac{ \gamma_i^2 }{ \sigma_i^2 } Z_{k,i} } + \frac{ \sum_{i=1}^N \frac{ \gamma_i^2 }{ \sigma_i^2 } Z_{k,i} }{2} - \sum_{i=1}^N \frac{ y_i \gamma_i }{ \sigma_i^2 }  Z_{k,i}.
\end{align*}
%\end{proof}
%
\subsection*{Proof of Proposition~\ref{prop population based full para gauss SSS}}
%
%\begin{proof}
%
\begin{align*}
& \mathrm{LLR}_k \\ 
& = \log \frac{ \max_{ (\alpha, \theta_{S_{k}} , \sigma^2 )  }  \Pi_{i=1}^N  f (y_{i} | Z_{k,i}; \alpha, \theta_{S_{k}} , \sigma^2 ) }{  \max_{ (\alpha, 0 , \sigma^2 )  }   \Pi_{i=1}^N   
f (y_{i} |  Z_{k,i};  \alpha, 0 , \sigma^2  ) } \\ 
& = \max_{ (\alpha, \theta_{S_{k}} , \sigma^2 )  } \left\{  \sum_{i=1}^N \log f (y_{i} | Z_{k,i};  \alpha, \theta_{S_{k}} , \sigma^2 )   \right\} - \max_{ (\alpha, 0 , \sigma^2 )  } \left\{ \sum_{i=1}^N \log f (y_{i} | Z_{k,i};  \alpha, 0 , \sigma^2 ) \right\} \\
& = \max_{ (\alpha, \theta_{S_{k}} , \sigma^2 )  } \left[ \sum_{i=1}^N  \log \frac{1}{ \sqrt{2 \pi \sigma^2} }  - \frac{ (y_i - \alpha - \theta_{S_k} Z_{ k,i } )^2 }{ 2 \sigma^2 }   \right] - \max_{ (\alpha, 0 , \sigma^2 )  } \left[ \sum_{i=1}^N   \log \frac{1}{ \sqrt{2 \pi \sigma^2} } - \frac{ (y_i - \alpha )^2 }{ 2 \sigma^2 }    \right] \\
&\mbox{The first term attains the maximum when } \\
& \hat{\alpha} = \frac{1}{N} \sum_{i=1}^N y_i - \frac{1}{N} \sum_{i=1}^N \widehat{\theta}_{S_k}  Z_{k,i} , \ \widehat{\theta}_{S_k} =  \frac{ \sum_{i=1}^N y_i Z_{k,i} - \hat{ \alpha } \sum_{i=1}^N Z_{k,i} }{ \sum_{i=1}^N Z_{k,i} } \mbox{, and } \widehat{\sigma^2} = \frac{ \sum_{i=1}^N (y_i - \hat{\alpha} - \hat{\theta}_{S_k} Z_{k,i} )^2 }{ N} ,\\
& \mbox{and the second term attains the maximum when } \tilde{\alpha} = \frac{1}{N} \sum_{i=1}^N y_i  \mbox{ and } \widetilde{\sigma^2} = \frac{ \sum_{i=1}^N (y_i - \tilde{\alpha} )^2 }{N}. \\
& = - \sum_{i=1}^N  \frac{1}{2} \log  \widehat{\sigma^2} -  \sum_{i=1}^N \frac{ (y_i - \hat{\alpha} - \hat{\theta}_{S_k} Z_{k,i} )^2 }{ 2  \widehat{ \sigma^2 } }  + \sum_{i=1}^N \frac{1}{2} \log  \widetilde{ \sigma ^2} + \sum_{i=1}^N \frac{ ( y_i - \tilde{\alpha} )^2 }{ 2  \widetilde{\sigma^2} } \\
& = \frac{N}{2} \log \left\{ \frac{ 1 }{ N } \sum_{i=1}^N  ( y_i - \frac{1}{N} \sum_{i=1}^N y_i )^2 \right\}  - \frac{N}{2} \log \widehat{ \sigma_{S_k}^2 } , \\
& \mbox{where } \widehat{  \sigma_{S_k}^2 } = \frac{1}{N} \left( \sum_{i=i}^N y_i^2   - \frac{ ( \sum_{i=1}^N y_i Z_{k,i}  )^2 }{ \sum_{i=1}^N Z_{k,i} } - \frac{ \left\{ \sum_{i=1}^N y_i (1- Z_{k,i} ) \right\}^2 }{ \sum_{i=1}^N (1 - Z_{k,i}) } \right) .
\end{align*}
%\end{proof}

%

%
\subsection*{Proof of Theorem~\ref{prop expectation based full para gauss SSS}}
%
%\begin{proof}
\begin{align*}
& \mathrm{LLR}_k  \\
& = \log \frac{ \max_{ (0, \theta_{S_{k}}, \sigma^2 )  }  \Pi_{i=1}^N  f (y_{i} | Z_{k,i}; 0, \theta_{S_{k}}, \sigma^2 ) }{  \max_{ (0, 0, \sigma^2 )  }  \Pi_{i=1}^N   
f (y_{i} |  Z_{k,i};  0 , 0, \sigma^2 ) } \\ 
& =  \max_{ ( 0, \theta_{S_k}, \sigma^2 ) } \left\{  \sum_{i=1}^N \log f (y_{i} | Z_{k,i};  0, \theta_{S_{k}} , \sigma^2 )   \right\} -   \max_{ ( 0, 0, \sigma^2 ) } \left\{ \sum_{i=1}^N \log f (y_{i} | Z_{k,i};  0, 0 , \sigma^2 ) \right\} \\
& =  \max_{ ( 0, \theta_{S_k} , \sigma^2 ) }  \left[ \sum_{i=1}^N  \log \frac{1}{ \sqrt{2 \pi \sigma^2} }  - \frac{ (y_i - \theta_{S_k} Z_{ k,i } )^2 }{ 2 \sigma^2 }   \right] - \max_{ ( 0, 0 , \sigma^2 ) } \left[ \sum_{i=1}^N   \log \frac{1}{ \sqrt{2 \pi \sigma^2} } - \frac{ y_i^2 }{ 2 \sigma^2 }    \right] \\
& \mbox{The first term attains the maximum when } \hat{\theta}_{S_k} = \frac{ \sum_{i=1}^N y_i Z_{k,i} }{ \sum_{i=1}^N Z_{k,i} } \mbox{ and } \widehat{\sigma^2} = \frac{ 1 }{ N } \sum_{i=1}^N  (y_i - \hat{\theta}_{S_k} Z_{k,i} )^2, \\
& \mbox{and the second term attains the maximum when } \tilde{\sigma^2} = \frac{ 1 }{ N } \sum_{i=1}^N y_i^2. \\
& = - \sum_{i=1}^N  \frac{1}{2} \log  \widehat{\sigma^2} -  \sum_{i=1}^N \frac{ (y_i - \hat{\theta}_{S_k} Z_{k,i} )^2 }{ 2  \widehat{ \sigma^2 } }  + \sum_{i=1}^N  \frac{1}{2} \log  \tilde{ \sigma ^2} + \sum_{i=1}^N \frac{ y_i^2 }{ 2  \tilde{\sigma^2} } \\
& = - \frac{N}{2} \log \left\{ \frac{1}{N} \sum_{i=1}^N  (y_i - \hat{\theta}_{S_k} Z_{k,i} )^2 \right\}   + \frac{N}{2} \log \left( \frac{1}{N} \sum_{i=1}^N y_i^2 \right) \\
& = - \frac{N}{2} \log  \left\{ \frac{1}{N} \sum_{i=1}^N y_i^2 - \frac{1}{N} \frac{ \left( \sum_{i=1}^N y_i   Z_{k,i} \right)^2  }{ \sum_{i=1}^N Z_{k,i} }  \right\}  + \frac{N}{2} \log \left( \frac{1}{N} \sum_{i=1}^N y_i^2 \right) .
\end{align*}
%\end{proof}

%
\subsection*{Proof of Theorem~\ref{new exp. bernoulli}}
%
%\begin{proof}
\begin{align*}
& \mathrm{LLR}_k  \\
& = \log \frac{ \max_{ (0, \theta_{S_{k}} )  }  \Pi_{i=1}^N  f (y_{i} | Z_{k,i}; 0, \theta_{S_{k}} ) }{    \Pi_{i=1}^N   
f (y_{i} |  Z_{k,i};  0 , 0 ) } \\ 
& =  \max_{ ( 0, \theta_{S_k} ) } \left\{  \sum_{i=1}^N \log f (y_{i} | Z_{k,i};  0, \theta_{S_{k}} )   \right\} - \sum_{i=1}^N \log f (y_{i} | Z_{k,i};  0, 0 ) \\
& = \max_{ ( 0, \theta_{S_k} ) } \left[ \sum_{i=1}^N \left\{ y_i \theta_{ S_k Z_{k,i} } - \log \left\{ 1 + \exp( \theta_{S_k} Z_{k,i} ) \right\} \right\}   \right] + N \log 2 \\
& \mbox{The first term attains the maximum when } \hat{\theta}_{S_k} = \log \frac{ \sum_{i=1}^N y_i Z_{k,i} }{ \sum_{i=1}^N ( 1 - y_i ) Z_{k,i}  }. \\
& = \left( \sum_{i=1}^N y_i Z_{k,i} \right) \log \frac{ \sum_{i=1}^N y_i Z_{k,i} }{ \sum_{i=1}^N ( 1 - y_i ) Z_{k,i}  } - \sum_{i=1}^N  \log \left\{ 1 + \exp \left( Z_{k,i} \log \frac{ \sum_{i=1}^N y_i Z_{k,i}  }{ \sum_{i=1}^N ( 1 - y_i ) Z_{k,i}  }  \right) \right\} + N \log 2 \\
& = \left( \sum_{i=1}^N y_i Z_{k,i} \right) \log \frac{ \sum_{i=1}^N y_i Z_{k,i} }{ \sum_{i=1}^N ( 1 - y_i ) Z_{k,i}  } - \sum_{i=1}^N Z_{k,i} \log \left\{ 1 + \frac{ \sum_{i=1}^N y_i Z_{k,i} }{ \sum_{i=1}^N ( 1 - y_i ) Z_{k,i}  } \right\} - \sum_{i=1}^N (1 - Z_{k,i}) \log 2 \\
& \qquad + N \log 2 \\
& = \left( \sum_{i=1}^N y_i Z_{k,i} \right) \log \frac{ \sum_{i=1}^N y_i Z_{k,i} }{ \sum_{i=1}^N ( 1 - y_i ) Z_{k,i}  } - \sum_{i=1}^N Z_{k,i} \log \left\{ \frac{ \sum_{i=1}^N Z_{k,i} }{ \sum_{i=1}^N ( 1 - y_i ) Z_{k,i}  } \right\} + \sum_{i=1}^N Z_{k,i} \log 2 \\
& = \left( \sum_{i=1}^N y_i Z_{k,i} \right) \log \frac{ \left( \sum_{i=1}^N y_i Z_{k,i} \right) }{ \left( \sum_{i=1}^N  Z_{k,i} \right) } + \left( \sum_{i=1}^N ( 1 - y_i ) Z_{k,i}  \right) \log \frac{ \sum_{i=1}^N ( 1 - y_i ) Z_{k,i}  }{ \sum_{i=1}^N  Z_{k,i}  } + \sum_{i=1}^N Z_{k,i} \log 2.
\end{align*}
%\end{proof}

\end{document}